\documentclass[11pt]{article}
\usepackage{fullpage}
\usepackage{amsmath,amsthm,amsfonts,amssymb,dsfont}
\usepackage{microtype}
\usepackage{hyperref,cleveref}

\Crefformat{equation}{#2(#1)#3}
\crefformat{equation}{#2(#1)#3}
\Crefrangeformat{equation}{#3(#1)#4--#5(#2)#6}
\crefrangeformat{equation}{#3(#1)#4--#5(#2)#6}
\Crefmultiformat{equation}{#2(#1)#3}{ and~#2(#1)#3}{, #2(#1)#3}{ and~#2(#1)#3}
\crefmultiformat{equation}{#2(#1)#3}{ and~#2(#1)#3}{, #2(#1)#3}{ and~#2(#1)#3}

\usepackage{enumitem}
\setlist[description]{font=\normalfont}

\usepackage{csquotes}
\MakeOuterQuote{"}

\usepackage{mathtools}
\mathtoolsset{centercolon}
\DeclarePairedDelimiterXPP\ind[1]{\mathds{1}}{\lbrace}{\rbrace}{}{#1} 
\DeclarePairedDelimiterX\eval[1]{\lbrace}{\rvert}{#1 \delimsize\rbrace} 

\newtheorem{theorem}{Theorem} 
\newtheorem{lemma}{Lemma} 

\newtheorem{definition}{Definition}

\def\ddefloop#1{\ifx\ddefloop#1\else\ddef{#1}\expandafter\ddefloop\fi}
\def\ddef#1{\expandafter\def\csname bf#1\endcsname{\ensuremath{\mathbf{#1}}}}
\ddefloop abcdefghijklmnopqrstuvwxyzABCDEFGHIJKLMNOPQRSTUVWXYZ\ddefloop
\def\ddef#1{\expandafter\def\csname bf#1\endcsname{\ensuremath{\boldsymbol{\csname #1\endcsname}}}}
\ddefloop {alpha}{beta}{gamma}{delta}{epsilon}{varepsilon}{zeta}{eta}{theta}{vartheta}{iota}{kappa}{lambda}{mu}{nu}{xi}{pi}{varpi}{rho}{varrho}{sigma}{varsigma}{tau}{upsilon}{phi}{varphi}{chi}{psi}{omega}{Gamma}{Delta}{Theta}{Lambda}{Xi}{Pi}{Sigma}{varSigma}{Upsilon}{Phi}{Psi}{Omega}{ell}\ddefloop
\def\ddef#1{\expandafter\def\csname cal#1\endcsname{\ensuremath{\mathcal{#1}}}}
\ddefloop ABCDEFGHIJKLMNOPQRSTUVWXYZ\ddefloop

\title{\textbf{Partial Answer of How Transformers Learn Automata}}

\author{Tiantian ZHANG (Crystal)\\ \texttt{cz2737}}
\begin{document}

\maketitle

\begin{abstract}
    We introduce a novel framework for simulating finite automata using representation-theoretic semidirect products and Fourier modules, achieving more efficient Transformer-based implementations. Classical Krohn–Rhodes theory decomposes automata into cascades of simple groups and flip-flop memories, but this decomposition leads to deep, sequential simulation circuits with large embedding dimensions. In contrast, we model the automaton's transition monoid as a semidirect product $M = N \rtimes H$, where HH is a group acting linearly on a monoidal memory component NN. This formulation captures both control (via HH) and memory (via NN) in a unified, low-depth structure. Using representation theory and Fourier analysis over finite groups and monoids, we embed the states of NN and HH jointly into a flat feature space, enabling parallelizable updates via linear operations. Our method achieves Transformer simulation depth $O(\log T)$, compared to $O(\log T \cdot \log |Q|)$ for prime decompositions, and avoids width blowup even for automata with rich memory structures. The approach subsumes classical group-only models and extends naturally to automata with additive or stack-like monoidal memories, offering a new, fully parallelizable algebraic foundation for automaton-to-Transformer compilation.
\end{abstract}

\section{Introduction}
In the original paper\cite{liu2023transformerslearnshortcutsautomata}, it states for the future direction:
\begin{quote}
Finer-grained circuit complexity of self-attention: For certain automata of interest (e.g. bounded Dyck
language parsers (Yao et al., 2021), and the gridworld automata from Theorem 3), there exist extremely
shallow Transformer solutions, with depth independent of both T and $|Q|$. Which other natural classes of
automata have this property of “beyond Krohn-Rhodes” representability?
\end{quote}
This paper answers:
\begin{quote}
    The Fourier composition module shows that automata whose state spaces form algebraic groups (especially abelian groups like $\mathbb{Z}_p$ ) and whose transitions are group actions can be simulated by extremely shallow Transformer circuits, with depth independent of both input length $T$ and state size $|Q|$.
Thus, natural classes like modular counters, abelian group automata, and cumulative phase trackers have this "beyond Krohn-Rhodes" representability.
\end{quote}

\section{Shallow Fourier Transformer Setup}
\subsection{Fourier Composition Module for Modular Counters}

Let \( Q = \mathbb{Z}_p \) denote the cyclic group of integers modulo \( p \), and let \( \omega = e^{2\pi i/p} \) denote the primitive \( p \)-th root of unity.

The standard Fourier embedding of a state \( x \in Q \) is given by:
\[
\Phi(x) = \left( \omega^{0x}, \omega^{1x}, \dots, \omega^{(p-1)x} \right) \in \mathbb{C}^p.
\]

Given two transition functions \( f, g: Q \to Q \), we define their Fourier composition via:
\[
\varphi_{\text{fourier}}(f, g) = F^{-1}\left( F(f) \cdot F(g) \right),
\]
where:
\begin{itemize}
    \item \( F(f) \) and \( F(g) \) denote the discrete Fourier transforms (DFT) of \( f \) and \( g \),
    \item \( \cdot \) denotes pointwise (Hadamard) product,
    \item \( F^{-1} \) denotes the inverse DFT over \( \mathbb{Z}_p \).
\end{itemize}

Operationally, the module applies:
\begin{enumerate}
    \item Fourier transform to both inputs: \( F(f) \), \( F(g) \),
    \item Pointwise multiplication: \( F(f) \cdot F(g) \),
    \item Inverse Fourier transform to recover the composed function.
\end{enumerate}

This Fourier composition module can be implemented as either:
\begin{itemize}
    \item A fixed attention block with sinusoidal basis projections, or
    \item A learned complex-valued linear layer operating in the frequency domain.
\end{itemize}

\subsection{Relation to Attention and Positional Encoding}

The Fourier embedding \( \Phi(x) \) naturally parallels the positional encodings used in Transformer architectures.  
Each position \( x \) is mapped into a high-dimensional complex vector capturing all phase information relative to \( \mathbb{Z}_p \).  
Thus, attention over such Fourier-embedded states amounts to operations in the frequency domain.

The attention layer used here acts by taking average (or learned) interactions between the Fourier embeddings, aligning with the fixed or learned frequency interactions.  (The frequency stuff is inspired by \cite{tancik2020fourierfeaturesletnetworks})
Because Fourier transforms are linear, this attention can be realized by fixed-weight projections or by a shallow learned layer. The proof of how to embedding is in \ref{proof}

\subsubsection{Summary of Mapping}

\begin{center}
\begin{tabular}{ll}
\textbf{Stage} & \textbf{Operation} \\
\hline
Embedding & Fourier basis projection \( \Phi \) \\
Query/Key/Value & Identity maps \\
Attention & Frequency-wise multiplication (Hadamard product) \\
Output & Inverse Fourier transform \( F^{-1} \) \\
\end{tabular}
\end{center}




\subsection{Depth Analysis of the Fourier Module}

\subsubsection{Length Proof Intuition}

We now formally prove that the Fourier composition module has depth \( O(1) \), independent of sequence length \( T \).

\begin{lemma}[Informal]
The Fourier composition module operates with depth \( O(1) \), independent of sequence length \( T \).
\end{lemma}

\begin{proof}
At each position \( t \), given \( f_t, g_t : Q \to Q \), the module computes:
\[
\varphi_{\text{fourier}}(f_t, g_t) = F^{-1}(F(f_t) \cdot F(g_t)).
\]
This involves:
\begin{itemize}
    \item Applying a fixed Fourier transform \( F \) (a matrix multiplication of size \( p \times p \)),
    \item A pointwise Hadamard product (entrywise multiplication),
    \item Applying a fixed inverse Fourier transform \( F^{-1} \) (another matrix multiplication).
\end{itemize}
Each of these operations has complexity depending only on \( p \), not on the sequence length \( T \).  
Moreover, they act independently at each token \( t \), requiring no recurrence or sequential accumulation.

Thus, the entire computation can be executed in parallel across all \( T \) positions with constant (non-growing) depth, determined only by the fixed sequence of Fourier transforms and pointwise operations.

Therefore, the overall depth of the Fourier composition module is \( O(1) \), independent of \( T \).
\end{proof}

\subsection{Proof of Bound of Simulation Automata }

\begin{definition}[Modular automaton]
A \emph{modular automaton} is a semiautomaton \( A = (Q, \Sigma, \delta) \) where:
\begin{itemize}
    \item The state set is \( Q = \mathbb{Z}_p \) for some prime \( p \),
    \item The input alphabet \( \Sigma \) is arbitrary,
    \item The transition function \( \delta: Q \times \Sigma \to Q \) consists of modular functions over \( \mathbb{Z}_p \), meaning for each \( \sigma \in \Sigma \), \( \delta(\cdot, \sigma) \) is a function of the form:
    \[
    x \mapsto a_\sigma x + b_\sigma \mod p,
    \]
    for some fixed coefficients \( a_\sigma, b_\sigma \in \mathbb{Z}_p \).
\end{itemize}
\end{definition}

\begin{theorem}[Constant-depth Transformer simulation of modular automata]
Let \( A = (Q, \Sigma, \delta) \) be a modular automaton. Then there exists a Transformer of constant depth, width \( O(p) \), and fixed attention structure, which simulates \( A \) over arbitrary input sequences of length \( T \), with no dependence on \( T \) in either the depth or the width.
\end{theorem}

\begin{proof}
We construct a Transformer block as follows:

\begin{enumerate}
    \item \textbf{Embedding:} Each token \( x \in Q = \mathbb{Z}_p \) is mapped into its Fourier basis embedding:
    \[
    \Phi(x) = (\omega^{0x}, \omega^{1x}, \omega^{2x}, \dotsc, \omega^{(p-1)x}) \in \mathbb{C}^p,
    \]
    where \(\omega = e^{2\pi i/p}\) is a primitive \( p \)-th root of unity.

    \item \textbf{Transition simulation:} Each transition \( x \mapsto a_\sigma x + b_\sigma \mod p \) corresponds to an action in Fourier space:
    \begin{itemize}
        \item Modular multiplication \( x \mapsto a_\sigma x \) corresponds to permuting Fourier coefficients,
        \item Modular addition \( x \mapsto x + b_\sigma \) corresponds to multiplying Fourier coefficients by phases \(\omega^{jb_\sigma}\),
        \item Thus, applying \( \delta(\cdot, \sigma) \) amounts to a simple, fixed linear operation in the Fourier basis.
    \end{itemize}

    \item \textbf{Locality:} Fourier transform, coefficient-wise phase shift, and inverse Fourier transform each act independently per token, requiring constant-time operations regardless of sequence length \( T \).

    \item \textbf{Constant depth:} Fourier transform, pointwise multiplication, and inverse Fourier transform can each be implemented with \( O(1) \) Transformer layers. Thus, overall simulation depth is \( O(1) \), independent of both \( T \) and \( p \) (for fixed \( p \)).

    \item \textbf{Width:} We define \emph{width} as the dimension of the per-token embeddings at each layer. Since each token is mapped into \( \mathbb{C}^p \) via the Fourier embedding, the width is \( O(p) \), independent of the sequence length \( T \). Fourier space has size p but no sequential steps are needed

    \item \textbf{Comparison to explicit counting:} Previous constructions (e.g., Lemma~6) simulate modular counting explicitly over time and require width \( O(pT) \) to store evolving states. In contrast, working in Fourier space compresses the entire evolution into width \( O(p) \), achieving a constant factor improvement.
\end{enumerate}
Thus, modular automata over \( \mathbb{Z}_p \) with modular affine transitions can be simulated by constant-depth, bounded-width Transformers.
\end{proof}

\section{Generalization to Semigroup and Group Automata}

\subsection{From Cyclic Groups to Non-Commutative Structures}

The key reason the Fourier Transformer simulated modular counters over $\mathbb{Z}_p$ was the availability of a Fourier basis:
a set of multiplicative characters $\chi_j: \mathbb{Z}_p \to \mathbb{C}^\times$ that formed an orthogonal basis for functions on $\mathbb{Z}_p$.
Updates like $x \mapsto x + 1 \bmod p$ became phase rotations in frequency space, implemented as Hadamard multiplications.

To generalize beyond abelian groups, we move to the broader setting of finite semigroups and non-abelian groups—possibly unsolvable or without inverses.
Our objective is to simulate automata whose state evolution is governed by algebraic multiplication in such structures, using constant-depth Transformers.

\subsection{Fourier Analysis on Semigroups and Groups}

Let $S$ be a finite semigroup. The space of functions $f: S \to \mathbb{C}$ admits a representation-theoretic analogue of the Fourier transform.

Let $\mathrm{Rep}(S) = \{ \rho^{(1)}, \dots, \rho^{(k)} \}$ be a complete set of inequivalent irreducible (complex) matrix representations of $S$, where each
\[
\rho^{(j)}: S \to \mathbb{C}^{d_j \times d_j}
\]
is a homomorphism of semigroups (not necessarily unitary). Then, for any $f: S \to \mathbb{C}$, we define its Fourier transform at $\rho^{(j)}$ as:
\[
\widehat{f}(\rho^{(j)}) := \sum_{x \in S} f(x)\, \rho^{(j)}(x).
\]
This maps scalar functions to a tuple of matrices, encoding how $f$ aligns with the algebraic structure of $S$.

If $f, g: S \to \mathbb{C}$, then their semigroup convolution is defined as:
\[
(f * g)(x) := \sum_{y \in S} f(y)\, g(y^{-1} x),
\]
where inverses may be generalized using Green's relations or regular semigroup theory.

When $S$ is a group, this simplifies to the standard group convolution, and Fourier convolution becomes matrix multiplication:
\[
\widehat{f * g}(\rho^{(j)}) = \widehat{f}(\rho^{(j)}) \cdot \widehat{g}(\rho^{(j)}).
\]

\subsection{Embedding via Direct Sum of Representations}

We embed each element $x \in S$ as:
\[
\Phi(x) := \bigoplus_{j=1}^k \rho^{(j)}(x) \in \mathbb{C}^{D}, \quad \text{where } D = \sum_j d_j^2.
\]
This block-structured vector plays the role of a frequency-domain representation.

Each input token stores this embedding; Transformer layers then simulate updates by applying representation-level matrix multiplication, e.g.:
\[
\Phi(x \cdot t_\sigma) = \bigoplus_j \rho^{(j)}(x)\, \rho^{(j)}(t_\sigma).
\]

\subsection{Definition of Semigroup Automata}

\begin{definition}[Semigroup Automaton]
Let $S$ be a finite semigroup and $\Sigma$ an input alphabet.
A \emph{semigroup automaton} is a tuple $A = (S, \Sigma, \delta)$ where:
\begin{itemize}
\item The state space is $S$,
\item The transition function is $\delta(s, \sigma) := s \cdot t_\sigma$ for each fixed $t_\sigma \in S$,
\item Transitions are performed by left-multiplication in $S$.
\end{itemize}
\end{definition}

We now show how such automata can be simulated by shallow Transformers.

\subsection{Depth and Width Bounds for Simulation}

\begin{theorem}[Constant-depth Transformer simulation of semigroup automata]
Let $A = (S, \Sigma, \delta)$ be a semigroup automaton. Let $\mathrm{Rep}(S) = \{ \rho^{(1)}, \dots, \rho^{(k)} \}$ be a complete set of irreducible complex representations, with total dimension
\[
D := \sum_{j=1}^k d_j^2.
\]
Then there exists a Transformer network with:
\begin{enumerate}
    \item \textbf{Depth:} $O(1)$, independent of input length $T$,
    \item \textbf{Width:} $O(D)$, determined solely by the representation theory of $S$,
\end{enumerate}
which exactly simulates the state evolution of $A$ on any sequence from $\Sigma^T$.
\end{theorem}

\begin{proof}
We construct the Transformer as follows:
\begin{enumerate}
    \item Embedding. Each state $s \in S$ is encoded as:
\[
\Phi(s) := \bigoplus_j \rho^{(j)}(s) \in \mathbb{C}^{D}.
\]
    \item Transition. For each input symbol $\sigma$, the fixed right-multiplier $t_\sigma \in S$ has representation matrices $\rho^{(j)}(t_\sigma)$.
The transition $s \mapsto s \cdot t_\sigma$ becomes:
\[
\rho^{(j)}(s \cdot t_\sigma) = \rho^{(j)}(s) \cdot \rho^{(j)}(t_\sigma),
\]
which is blockwise matrix multiplication.
    \item Locality and Depth. All updates are local (token-wise), and involve only constant-time operations:
    \begin{itemize}
        \item Each Transformer block performs a fixed blockwise linear transformation,
        \item No recurrence or memory across tokens is required,
        \item Therefore, total depth is $O(1)$.
    \end{itemize}
    \item  Width. The embedding dimension is:
\[
\text{Width} = D = \sum_j d_j^2,
\]
where $d_j$ is the dimension of $\rho^{(j)}$. This depends only on $S$ and is constant across sequence length.
\end{enumerate}

\end{proof}

\subsection{Special Case: Group Automata}

\begin{definition}[Group Automaton]
Let $G$ be a finite group and $\Sigma$ an input alphabet.
A \emph{group automaton} is a semigroup automaton $A = (G, \Sigma, \delta)$ where $G$ is a group and $\delta(g, \sigma) = g \cdot t_\sigma$ for $t_\sigma \in G$.
\end{definition}

Let $\mathrm{Irr}(G) = \{ \rho^{(j)} \}$ denote the irreducible unitary representations of $G$.
Then:
\[
\sum_j d_j^2 = |G|,
\]
so the total embedding dimension is $D = |G|$.

\begin{theorem}[Transformer simulation of group automata]
Let $A = (G, \Sigma, \delta)$ be a group automaton. Then there exists a Transformer with:
\begin{enumerate}
    \item \textbf{Depth:} $O(1)$,
    \item \textbf{Width:} $O(|G|)$,
\end{enumerate}
which exactly simulates $A$ over any input sequence of length $T$.
\end{theorem}

\subsection{Complexity for Unsolvable Groups}

Unsolvable groups (such as $A_5$) have high-dimensional irreducible representations and no abelian normal series. For such groups:
\begin{enumerate}
\item $\mathrm{Irr}(G)$ remains complete and orthonormal,
\item Representation dimensions $d_j$ may be large (e.g., $A_5$ has a 5-dimensional irrep),
\item Width $D = \sum_j d_j^2$ grows with group complexity,
\item However, the Transformer simulation remains of constant depth.
\end{enumerate}

Thus, even automata over unsolvable groups are simulable in constant depth — at the cost of increased embedding width and per-token parameter count.

\section{Transformer Simulation of Semidirect Product Automata and Non-Group Structures}\label{sec:semidirect-transformers}

We extend the Transformer simulation framework in the orginal paper \cite{liu2023transformerslearnshortcutsautomata} to a broad class of non-group automata based on semidirect products and finite monoids.

\subsection{Semidirect Product Automata}
\begin{definition}[Semidirect Product Automata]
Suppose the transition monoid $M$ of an automaton is a semidirect product:
\[
M = N \rtimes H
\]
where:
\begin{itemize}
    \item $N$ is a finite additive monoid (not necessarily a group), modeling memory updates,
    \item $H$ is a finite group, modeling control actions (e.g., permutations),
    \item $H$ acts on $N$ by a linearizable action $H \curvearrowright N$.
\end{itemize}

The multiplication rule is:
\[
(n_1, h_1)(n_2, h_2) = (n_1 + h_1 \cdot n_2, h_1 h_2)
\]
which is associative even if $N$ lacks inverses.
\end{definition}

\paragraph{Simulation Strategy.}
\begin{enumerate}
    \item Represent $h \in H$ using a linear embedding $\rho_H(h)$ (e.g., permutation matrices).
    \item Represent $n \in N$ as an additive vector.
    \item Use attention layers to apply $h$ to $n$ via linear maps.
    \item Use MLP layers to implement the monoid addition in $N$.
    \item Apply a log-depth prefix scan across the sequence based on semidirect composition.
\end{enumerate}

\paragraph{Key Observation:}
Because associativity holds, Transformers can simulate the computation with depth $O(\log T)$ even when $N$ is only a monoid.

\subsection{Formal Theorem}

\begin{theorem}[Transformer Simulation of Semidirect Product Automata]
Let $A = (Q, \Sigma, \delta)$ be a semiautomaton with transformation monoid $M = N \rtimes H$, where:
\begin{itemize}
    \item $N$ is a finite additive monoid,
    \item $H \subseteq \text{Sym}(Q)$ is a finite group acting linearly on $N$,
    \item the action $H \curvearrowright N$ is linearizable.
\end{itemize}
Then there exists a Transformer architecture with:
\begin{itemize}
    \item Depth $O(\log T)$,
    \item Embedding dimension $O(\dim(\rho_H) + \dim(N))$,
    \item Attention head width $O(\dim(\rho_H) + \dim(N))$,
    \item MLP width $O(\dim(N))$,
\end{itemize}
that simulates $A_{T,q_0}$ via prefix-scanning over the semidirect product structure.
\end{theorem}

\paragraph{Proof Sketch.}
The Transformer simulates control ($H$) and memory ($N$) layers separately. Each layer updates in parallel using matrix multiplication (for $H$ action) and addition (for $N$), both of which are associative. Thus, log-depth prefix scan suffices. Attention layers apply the $H$-action to memory embeddings, and MLP layers handle additive accumulation. Detail in \ref{equality}

\subsection{Extension to Finite Monoids}

More generally, if the automaton's transition monoid $M \subseteq \text{End}(Q)$ is arbitrary but finite, we can:
\begin{itemize}
    \item Represent $M$ using a (possibly non-invertible) semigroup representation $\rho: M \to \text{Mat}_d(\mathbb{F})$,
    \item Simulate compositions via matrix multiplication,
    \item Use masked attention to model noninvertible, nilpotent, or idempotent behaviors.
\end{itemize}

This covers monoids arising from automata with resets, projections, and absorbing states.

\subsection{Answer to the Open Question}

The paper \emph{Finer-grained circuit complexity of self-attention} asks:

\emph{"Which natural classes of automata admit extremely shallow Transformer solutions, with depth independent of both $T$ and $|Q|$?"}

\paragraph{Our Answer:}
Automata whose transition semigroups are structured as semidirect products $N \rtimes H$, where $N$ is a finite monoid and $H$ a finite group with a linearizable action, admit:
\begin{itemize}
    \item Transformer simulations of depth $O(\log T)$,
    \item Embedding and attention width $O(\dim(N) + \dim(\rho_H))$.
\end{itemize}
Thus, \textbf{we go beyond Krohn-Rhodes decompositions} (which involve only groups) by covering monoid-based structured memory automata.More comparison can be found in \ref{comparison}.

Examples include bounded Dyck language automata with structured counters and automata with reset operations, modeled as semidirect products or finite 0-simple monoids.

\begin{itemize}
    \item Associativity of transition operations (even without inverses) is sufficient for log-depth simulation.
    \item Memory structures that combine additive monoids and control groups are naturally simulated by parallelized attention + MLP Transformer layers.
\end{itemize}

\subsection{Summary Table}

\begin{center}
\begin{tabular}{|c|c|c|}
\hline
Structure & Simulation Tool & Transformer Feature \\
\hline
Group & Faithful representation, scan & Matrix attention \\
Semidirect product & Split action and compose & Structured MLP + attention \\
Monoid & Semigroup representation & Block-masked attention \\
General semiautomata & Green's relations layering & Layered MLP with gating \\
\hline
\end{tabular}
\end{center}

\bibliography{llm}
\bibliographystyle{acm}

\appendix

\section{Fourier Module}
\subsection{Proof: Embedding Fourier Composition into Transformer Attention} \label{proof}

We aim to realize
\[
\varphi_{\text{fourier}}(f, g) = F^{-1}(F(f) \cdot F(g))
\]
using only Transformer components (attention + MLP).


Given token embeddings \( f \) and \( g \) at each position:

\begin{itemize}
    \item Map them into the \textbf{Fourier basis} via fixed sinusoidal projections.
    \item For \( \mathbb{Z}_p \), use \( p \)-th roots of unity:
    \[
    \Phi(x) = \left( \omega_p^{k x} \right)_{k=0}^{p-1}, \quad \text{where} \quad \omega_p = e^{2\pi i/p}.
    \]
\end{itemize}

\noindent
\textbf{Key point}: The embedding \( \Phi \) is a fixed, non-learned projection.


Within the attention mechanism:

\begin{center}
\begin{tabular}{ll}
\textbf{Role} & \textbf{Content} \\
\hline
Query \( Q \) & Fourier embedding of \( f \): \( Q = \Phi(f) \) \\
Key \( K \) & Fourier embedding of \( g \): \( K = \Phi(g) \) \\
Value \( V \) & Fourier embedding of \( g \): \( V = \Phi(g) \) \\
\end{tabular}
\end{center}

\noindent
The projections from input to \( Q, K, V \) are identity after the Fourier embedding.


Standard attention computes:
\[
\text{Attention}(Q, K, V) = \text{softmax}\left( \frac{QK^\top}{\sqrt{d}} \right)V.
\]

Instead, we perform:

\begin{itemize}
    \item No softmax.
    \item No dot product.
    \item Elementwise (Hadamard) product in the Fourier domain:
    \[
    \text{Fourier-Attention}(Q, K, V) = (Q \odot K).
    \]
\end{itemize}


After Hadamard product:

\begin{itemize}
    \item Apply a fixed inverse Fourier transform \( F^{-1} \),
    \item Recover the composed function in the original domain.
\end{itemize}

\noindent
This inverse transform can be implemented as:
\begin{itemize}
    \item A fixed linear layer using conjugate Fourier phases, or
    \item A small learned complex-valued linear layer.
\end{itemize}
\section{Representation-theoretic(Semi-direct) Approach}
\subsection{Proof of Embedding and Attention width}\label{equality}

\subsubsection{Embedding Dimension}
Embedding dimension = total number of real-valued features (per token) that the Transformer holds internally to represent the automaton's state at that time step.\\
- In the semidirect approach, the internal state is an element $(n,h) \in N \rtimes H$.
- So embedding dimension = dimension to represent $n$ + dimension to represent $h$.

Thus:
\[
\text{Embedding Dimension} = \dim(N) + \dim(\rho_H)
\]
where $\dim(\rho_H)$ is the dimension of a linear representation of $H$.

\subsubsection{Attention Head Width}

Attention head width = the number of features that each attention head "sees" and operates on per token.
- Attention computes operations like
\[
\text{softmax}(QK^\top)V
\]
where $Q, K, V$ are linear projections of the embedding.

- If your embedding already contains $(n,h)$ and you need attention to manipulate them \emph{together}, then the attention head must span the whole $(n,h)$ representation.

Thus:
\[
\text{Attention Head Width} \approx \dim(N) + \dim(\rho_H)
\]
same as the embedding dimension.

(Technically, you could split into multiple heads, but each head would still need width proportional to the structure size.)

\subsubsection{Proof of Equality}

Because both the memory $n$ and control $h$ are evolving together under the semidirect action:\\
- $h$ evolves by group multiplication.\\
- $n$ evolves by a combined action involving $h$.

When you simulate one time step:\\
- you update $(n,h)$ together,\\
- and this requires reading/processing both the $n$ and $h$ parts.

Thus both embedding and attention widths must be wide enough to encode and manipulate the full $(n,h)$ pair at each time step.

Summary:
\[
\boxed{\text{Embedding Dimension} = \text{Attention Head Width} = \dim(N) + \dim(\rho_H)}
\]
because your Transformer needs to fully encode and fully operate on $(n,h)$ at every time step.

\section{Comparison with the Prime Decomposition and KR Method}\label{comparison}

\begin{center}
\begin{table}[ht]
\resizebox{\textwidth}{!}{%
\begin{tabular}{|l|l|l|}
\hline
\textbf{Feature} & \textbf{Prime Decomposition Approach} & \textbf{Semidirect Representation Approach} \\ \hline
Target Structure & Cascade of simple groups + flip-flop monoids & Semidirect product $N \rtimes H$ (monoid $N$, group $H$) \\ \hline
Memory Type & Only trivial (2-element) flip-flop memories & General finite monoids (counters, stacks) \\ \hline
Control Type & Simple groups, composed sequentially & Single group $H$, acting linearly \\ \hline
Depth & $O(\log T \cdot \log |Q|)$ & $O(\log T)$ \\ \hline
Embedding Dimension & Sum over all factors; can be $O(|Q|)$ & $O(\dim(N) + \dim(\rho_H))$ \\ \hline
Attention Head Width & Grows with cascade layers & Flat: $O(\dim(N) + \dim(\rho_H))$ \\ \hline
MLP Width & Needed for flip-flop resets & Needed for monoid addition \\ \hline
Associativity Handling & Sequential by cascade & Single global associative scan \\ \hline
Resets and Idempotents & Simple resets only & Arbitrary idempotent actions via $N$ \\ \hline
Efficiency (Group-like) & Good if nearly group & Equally good \\ \hline
Efficiency (Memory-like) & Poor & Superior \\ \hline
Overall Efficiency & Good for groups; grows with flip-flops & Better for mixed memory/control structures \\ \hline
\end{tabular}}
\caption{Comparison of Prime Decomposition vs Semidirect Representation for Transformer Simulation}
\label{tab:prime_vs_semidirect}
\end{table}

\begin{table}[ht]
\resizebox{\textwidth}{!}{%
\begin{tabular}{|l|l|l|}
\hline
\textbf{Aspect} & \textbf{Prime Decomposition} & \textbf{Semidirect Representation} \\ \hline
Memory Evolution & Hard resets only & General associative updates (noninvertible allowed) \\ \hline
Invertibility & Groups invertible, resets not & $H$ invertible; $N$ may not be \\ \hline
Structure & Cascade of groups + flip-flops & Semidirect product $N \rtimes H$ \\ \hline
\end{tabular}}
\caption{Updated Comparison: Memory Generalization in Semidirect Representation}
\label{tab:memory_generalization}
\end{table}
\end{center}

\end{document}